\documentclass[conference]{IEEEtran}
\makeatletter
\def\ps@IEEEtitlepagestyle{%
  \def\@oddfoot{\smash{\parbox[b][0pt][b]{\textwidth}{\footnotesize\raggedright
    Accepted for publication at the Edge Intelligence Workshop, ACM/IEEE
    Symposium on Edge Computing (SEC) 2026. \copyright~2026 IEEE. Personal use
    of this material is permitted. Permission from IEEE must be obtained for
    all other uses.}}}%
  \def\@evenfoot{}\def\@oddhead{}\def\@evenhead{}}
\makeatother

\usepackage{cite}
\usepackage{amsmath,amssymb,amsfonts}
\usepackage{amsthm}
\usepackage{mathtools}
\usepackage{graphicx}
\usepackage{booktabs}
\usepackage{xcolor}
\usepackage{hyperref}
\usepackage[capitalize,noabbrev]{cleveref}

\let\citep\cite
\let\citet\cite

\newcommand{\EE}{\mathbb{E}}
\newcommand{\tdraft}{t_{d}}
\newcommand{\tverify}{t_{v}}
\newcommand{\tar}{t_{\mathrm{ar}}}
\newcommand{\rtt}{\mathrm{RTT}}

\theoremstyle{plain}
\newtheorem{proposition}{Proposition}

\theoremstyle{remark}
\newtheorem{remark}[proposition]{Remark}

\begin{document}

\title{Speculation at a Distance: Where Edge-Cloud Speculative Decoding Actually Pays Off}

\author{\IEEEauthorblockN{Yuan Lyu, Bharath Irukulapati, and Jaya Prakash Champati}
\IEEEauthorblockA{\textit{Department of Computer Science} \\
\textit{University of Victoria} \\
Victoria, BC, Canada \\
yuanlyu@uvic.ca, bharath.irukulapati@gmail.com, jpchampati@uvic.ca}
}

\maketitle

\begin{abstract}
Speculative decoding (SD) accelerates LLM inference by $1.5$--$3\times$ when the
draft and target models are co-located. 
This has motivated a distributed variant (DSD) that places the draft model on an
edge device while the target stays in the cloud.
We show with closed-form expressions that DSD's per-request latency benefit
is limited in the WAN edge-cloud settings analyzed here. If the server can host
both models, co-located SD has lower latency and communication than synchronous
DSD, with the same per-output FLOPs and model-weight memory.
Pipelining can make DSD competitive with co-located SD only when communication
fits within the edge drafting time window and wasted pre-drafting remains
sufficiently small; the representative WAN settings analyzed here do not
satisfy this condition.
Against cloud autoregressive decoding, DSD can reduce latency only inside a
bounded window given the target-model speed, acceptance rate, drafting time, and
communication overhead.
DSD is also infeasible against closed-source APIs without a verifier-only interface. 
The main case for DSD appears in multi-tenant capacity.
Under cross-client overlap, offloading draft compute lets a saturated cloud
server sustain $(1 + \gamma\,t_d/t_v)\times$ more concurrent clients at the same per-client
rate, where $\gamma$ is the speculation length and $t_d, t_v$ are the per-step draft
and verification times. DSD should therefore be evaluated primarily by multi-tenant
capacity and server throughput, not only by single-request latency.
\end{abstract}


\section{Introduction}
\label{sec:intro}
Speculative decoding (SD) has become a standard tool for accelerating large
language model (LLM) inference \citep{leviathan2023fast,chen2023accelerating}.
The standard \textit{co-located SD}, which places a small \emph{draft} model and
a large \emph{target} model on the same hardware, exploits the asymmetry between
sequential token generation and parallel verification to deliver
$1.5$--$3\times$ wall-clock speedups~\citep{xia2024survey} while guaranteeing the
tokens are generated according to the target distribution~\citep{leviathan2023fast}.

Driven by the proliferation of capable small language models on edge devices and
their on-device performance, a growing body of work proposes a
\emph{distributed} variant of SD (DSD): the draft models run on edge devices
such as smartphones, gateways, or on-premise servers, while the target models run in data
centers, with each verification round involving communication over a network
\citep{ning2025dssd,zhang2026picospec,yu2025dsd,li2025sled,koh2025resource,zheng2025tslt}.

Existing DSD systems each attack a different bottleneck in the system, but share a common
evaluation gap. DSSD~\citep{ning2025dssd} splits verification between edge
and cloud, sending the vocabulary distribution on the downlink only when a draft
is rejected rather than on every uplink round; TSLT~\citep{zheng2025tslt}
sends only a truncated candidate
set of logits on the uplink; PicoSpec~\citep{zhang2026picospec} applies
asynchronous pipelining to overlap edge drafting with cloud verification;
DSD~\citep{yu2025dsd} applies an Adaptive Window Control policy that
dynamically adjusts the speculation window; SLED~\citep{li2025sled} batches
verification across multiple edge clients to raise server throughput;
SpecEdge~\citep{park2025specedge} disaggregates drafting to
consumer-grade GPUs on edge devices, hiding communication and computation
latency by overlapping proactive edge-side drafting and server-side
pipeline-aware scheduling; PipeSD~\citep{han2026pipesd} adds token-batch
pipeline scheduling and adaptive verification triggering; and
Koh et al.~\citep{koh2025resource} propose resource-aware parallel
speculative decoding across heterogeneous devices. 
Reported gains span $1.1$--$2.9\times$ latency speedups in DSSD, PicoSpec, PipeSD, and DSD, or
$2.2\times$ aggregate throughput in SLED and SpecEdge. 
Most of the above works measured their proposed DSD method against
\textit{cloud autoregressive decoding (AR)} or against other DSD methods; see Table~\ref{tab:baselines}.
Only DSD~\citep{yu2025dsd} and SpecEdge~\citep{park2025specedge} compare against
co-located SD, and in both cases the distributed advantage disappears once the
round-trip time (RTT) grows. We list the thresholds in Table~\ref{tab:baselines}
and analyze them in \S\ref{sec:pipelining}.

\begin{table}[t]
\caption{Baselines used in the representative DSD systems. Only DSD and
SpecEdge evaluate against co-located SD (in bold).}
\label{tab:baselines}
\centering
\footnotesize
\setlength{\tabcolsep}{3pt}
\begin{tabular*}{\columnwidth}{@{\extracolsep{\fill}}lcccc@{}}
\toprule
System & Cloud AR & Co-located\ SD & Other DSD & Layer split\textsuperscript{$*$} \\
\midrule
DSSD~\citep{ning2025dssd}          & $\checkmark$ & $\times$ & $\checkmark$ & $\times$ \\
TSLT~\citep{zheng2025tslt}         & $\checkmark$ & $\times$ & $\checkmark$ & $\times$ \\
PicoSpec~\citep{zhang2026picospec} & $\checkmark$ & $\times$ & $\checkmark$ & $\checkmark$ \\
SLED~\citep{li2025sled}            & $\checkmark$ & $\times$ & $\times$ & $\times$ \\
PipeSD~\citep{han2026pipesd}       & $\times$ & $\times$ & $\checkmark$ & $\times$ \\
Koh et al.~\citep{koh2025resource}\,\textsuperscript{\dag} & $\times$ & $\times$ & $\times$ & $\times$ \\
DSD~\citep{yu2025dsd}              & $\times$ & $\pmb{\checkmark}$\textsuperscript{\ddag} & $\checkmark$ & $\times$ \\
SpecEdge~\citep{park2025specedge}  & $\checkmark$ & $\pmb{\checkmark}$\textsuperscript{\S} & $\times$ & $\checkmark$ \\
\bottomrule
\end{tabular*}

\par\smallskip
\begin{minipage}{\columnwidth}
\footnotesize\raggedright
\textsuperscript{$*$}\,Layer split: Part of the target model offloaded to the edge.\\[2pt]
\textsuperscript{\dag}\,Koh et al.\ compare against resource-allocation
heuristics rather than against a decoding method.\\
\textsuperscript{\ddag}\,DSD reports that distributed execution wins only
below $60$~ms RTT.\\
\textsuperscript{\S}\,SpecEdge measures a latency advantage over co-located
SD only when RTT is below $50$~ms.
\end{minipage}
\end{table}

We study the DSD system in the edge-cloud WAN setting, with a wide-area network
between the edge draft and the cloud target.
Our closed-form analysis in \S\ref{sec:position} and \S\ref{sec:pipelining}
separates synchronous and asynchronous DSD. For synchronous systems such as
DSSD~\citep{ning2025dssd} and DSD~\citep{yu2025dsd}, per-request latency gains
over cloud AR exist only inside a bounded RTT window. The tolerable RTT is
essentially the expected per-round target model processing time,
minus a few fixed compute overheads. This implies that, for the illustrative parameter settings
in Table~\ref{tab:scenarios}, the resulting RTT budget excludes representative
4G/LTE and cross-region RTTs for several target speeds and acceptance rates.
For asynchronous systems such as PicoSpec~\citep{zhang2026picospec},
pipelining can mask communication only while the relevant network and
verification work fits within the edge drafting window.
Section~\ref{sec:pipelining} compares pipelined DSD with co-located SD.
Under the low-transmission-overhead model, pipelined DSD is faster 
only when the round trip time is less than the edge drafting time per round.
Among the systems and operating conditions surveyed here, no existing DSD system
demonstrates a latency advantage over co-located SD when the RTT exceeds its
edge drafting window, which is around $50$~ms in the reported configurations.
The pipelined DSD systems that report gains at higher RTT measure against cloud AR or
other DSD systems, but never against co-located SD.
 
DSD does, however, retain one advantage that our analysis makes precise: under
multi-tenant serving, 
offloading draft compute to the edge lets the cloud sustain more concurrent
clients at the same per-client output token rate than co-located SD.
Our contribution is to make that distinction explicit against the co-located SD
baseline: DSD's capacity case is stronger than its per-request latency case,
whose gains depend on RTT and the baseline.

Below are our main contributions.
\begin{itemize}
\item A closed-form per-request analysis in~\S\ref{sec:position} showing that
co-located SD has lower latency and communication than synchronous DSD.
Section~\ref{sec:pipelining} extends the latency comparison to pipelined DSD
in the analyzed WAN regime. Against cloud AR, DSD reduces latency only within
a bounded RTT window.
\item A multi-tenant capacity result, Prop.~\ref{prop:throughput}, that
quantifies DSD's capacity benefit---$(1 + \gamma\,\tdraft/\tverify)\times$ more
concurrent clients per cloud server---and reframes DSD as a system-capacity
technique.
\item Result-by-result support from published measurements:
the $350{+}$ co-located SD experiments of Yan et al.~\citep{yan2025decoding} validate our 
effective-time model;
DSSD's network results trace the cloud-AR break-even window of
Prop.~\ref{prop:lat-ar};
SLED~\citep{li2025sled} and SpecEdge~\citep{park2025specedge} support DSD's
server-capacity benefit, while CoSine~\citep{gao2025cosine}
supports the same resource-allocation principle in a datacenter setting.
\end{itemize}

Table~\ref{tab:comparison} summarizes the closed-form comparisons in
\S\ref{sec:position}.

\begin{table}[t]
\caption{Analytical comparison summarizing Prop.~\ref{prop:coloc-dom},
Prop.~\ref{prop:lat-ar}, Rem.~\ref{rem:api-ar}, and
Prop.~\ref{prop:throughput}. Arrows indicate the favorable direction; \textbf{bold}
marks the best entry in each row. 
}
\label{tab:comparison}
\centering
\small
\setlength{\tabcolsep}{5pt}
\begin{tabular}{@{}lccc@{}}
\toprule
Metric & Cloud AR & Co-located SD & DSD \\
\midrule
\multicolumn{4}{@{}l}{\emph{Per-request / configuration}} \\
Per-token latency $\downarrow$        & Medium       & \textbf{Low}  & High\textsuperscript{\dag} \\
Per-token server time $\downarrow$    & High         & Medium        & \textbf{Low} \\
Total memory $\downarrow$             & \textbf{Low} & Medium        & Medium \\
Network per round $\downarrow$        & Low          & \textbf{None} & High \\
Closed-API support                    & $\checkmark$ & $\checkmark$  & $\times$\textsuperscript{\ddag} \\
\midrule
\multicolumn{4}{@{}l}{\emph{Multi-tenant}} \\
Concurrent clients $\uparrow$         & ---$^*$ & Baseline & \textbf{High}\textsuperscript{\S} \\
\bottomrule
\end{tabular}
\par\smallskip
\footnotesize
\textsuperscript{\dag}\,Medium latency only inside the narrow RTT window of
Prop.~\ref{prop:lat-ar}. \quad
\textsuperscript{\ddag}\,See Remark~\ref{rem:api-ar}.
\textsuperscript{\S}\,$(1{+}\gamma\tdraft/\tverify)\times$ co-located SD under cross-client overlap; see Prop.~\ref{prop:throughput}.\quad
$^*$\,The capacity comparison concerns where draft compute runs;
cloud AR has no draft model and lies outside this comparison.
\end{table}

\section{Background and Setting}
\label{sec:background}
We compare three configurations, all serving the same target model $M_t$. The
simplest is \emph{cloud AR}: the target runs alone on the cloud server with no
drafting and no speculation, generating one token per forward pass. Let $\tar$
denote its per-token wall-clock cost. Table~\ref{tab:notation} collects the notation used
throughout.

\begin{table}[t]
\caption{Notation used throughout the paper.}
\label{tab:notation}
\centering
\footnotesize
\setlength{\tabcolsep}{4pt}
\begin{tabular}{@{}ll@{}}
\toprule
Symbol & Meaning \\
\midrule
\multicolumn{2}{@{}l}{\emph{Models and distributions}}\\
$M_t$, $M_s$ & target model and small draft model \\
$|M_t|$, $|M_s|$ & weight-memory footprints in bytes \\
$p(\cdot\mid x_{<i})$ & target next-token distribution \\
$q(\cdot\mid x_{<i})$ & draft next-token distribution \\
$|V|$ & vocabulary size \\
\midrule
\multicolumn{2}{@{}l}{\emph{Speculation}}\\
$\gamma$ & speculation length in draft tokens per round \\
$x_1,\dots,x_\gamma$ & proposed draft tokens \\
$r_i$ & per-token acceptance prob.\ $\min(1,p(x_i)/q(x_i))$ \\
$\alpha$ & per-position acceptance probability~\eqref{eq:alpha} \\
$A$ & output tokens per round, $A\in\{1,\dots,\gamma{+}1\}$ \\
$\EE[A]$ & expected output tokens per round~\eqref{eq:expected} \\
\midrule
\multicolumn{2}{@{}l}{\emph{Per-token / per-round timing}}\\
$\tar$ & cloud-AR per-token wall-clock time \\
$\tdraft$ & time per draft token \\
$\tverify$ & time for one forward pass verifying $\gamma$ tokens \\
$T_{\text{eff}}^{\text{coloc}},T_{\text{eff}}^{\text{dsd}},T_{\text{eff}}^{\text{pipe}}$
& mean wall time per output token~\eqref{eq:coloc_eff},\eqref{eq:dist_eff},\eqref{eq:pipe_eff} \\
$T_{\mathrm{round}}^{(\cdot)}$ & per-round wall time, $T_{\text{eff}}^{(\cdot)}\EE[A]$ \\
$w$ & speculative-waste fraction under pipelining \\
\midrule
\multicolumn{2}{@{}l}{\emph{Network}}\\
$\rtt$ & payload-independent round-trip time \\
$T_{\mathrm{tx}}(\gamma)$ & transmission time, $\gamma b/R$~\eqref{eq:ttx} \\
$b$, $b_{\mathrm{prob}}$ & per-draft payload, bytes per prob.\ value \\
$R$ & link bandwidth \\
$\rtt_{\max}$ & break-even network budget~\eqref{eq:rtt_max} \\
\midrule
\multicolumn{2}{@{}l}{\emph{Compute (FLOPs) and server throughput}}\\
$C_{\mathrm{AR}},C_{\mathrm{draft}},C_{\mathrm{verify}}$ & per-token / per-round FLOPs \\
$c$ & draft-to-AR FLOP ratio $C_{\mathrm{draft}}/C_{\mathrm{AR}}$ \\
$r$ & required long-run output rate per admitted client \\
$N_{\mathrm{ar}},N_{\mathrm{coloc}},N_{\mathrm{dsd}}$ & maximum clients sustainable at rate $r$ \\
\midrule
\multicolumn{2}{@{}l}{\emph{Hypothetical API cost; see Rem.~\ref{prop:hypo-api}}}\\
$p_{\mathrm{in}},p_{\mathrm{out}}$ & per-token input / output price \\
$F_{\mathrm{ver}}(\gamma)$ & per-call verification fee \\
\bottomrule
\end{tabular}
\end{table}

\subsection{Co-located SD} A small draft model $M_s$ and large target model
$M_t$ have weight-memory footprints $|M_s|$ and $|M_t|$ bytes and token
distributions $q(\cdot \mid x_{<i})$ and $p(\cdot \mid x_{<i})$ respectively. The draft model proposes $\gamma$ tokens
$x_1,\ldots,x_\gamma \sim q$; the target model verifies them in one forward pass,
accepting each $x_i$ with probability
$r_i = \min(1, p(x_i)/q(x_i))$. The \emph{per-position acceptance probability}
$\alpha$ is the expectation of this acceptance rule over the draft distribution:

\begin{equation}
    \alpha \;\triangleq\; \EE_{x \sim q}\!\left[\min\!\left(1,\, \tfrac{p(x)}{q(x)}\right)\right]
    \;=\; \sum_{x} \min\!\big(q(x),\, p(x)\big) 
    \label{eq:alpha}
\end{equation}
Following \citet{leviathan2023fast}, we use a constant conditional
acceptance model. For every draft position $i$, the probability of accepting
$x_i$, conditioned on accepting $x_1,\ldots,x_{i-1}$, is $\alpha$, where
$0\le\alpha\le1$.

Let $A \in \{1,\ldots,\gamma+1\}$ be the number of tokens produced per round.
It consists of $A-1$ accepted draft tokens plus one additional target token.
The additional token is a correction after the first rejection or a bonus token
when all $\gamma$ drafts are accepted. The event $\{A\ge a\}$ occurs exactly
when the first $a-1$ drafts are accepted. Therefore,
\begin{equation}
    P(A\ge a)=\alpha^{a-1}.
    \label{eq:tail}
\end{equation}
Applying the tail-sum formula gives
\begin{equation}
\EE[A]=\sum_{k=0}^{\gamma}\alpha^k
=\begin{cases}
(1-\alpha^{\gamma+1})/(1-\alpha), & 0\le\alpha<1,\\
\gamma+1, & \alpha=1.
\end{cases}
\label{eq:expected}
\end{equation}

Let $\tdraft$ denote the time per draft token and $\tverify$ the time to verify
the $\gamma$ tokens in parallel. To isolate placement and communication, the
comparisons below use the same fixed, contention-free $\tdraft$ and $\tverify$
for the same models in both placements. If measured service times differ by
placement, the corresponding configuration-specific values must be substituted.
For any configuration $(\cdot)$, we define $T_{\mathrm{eff}}^{(\cdot)}$ as its
expected wall-clock time per generated output token: the wall-clock time of one
decoding round divided by the expected number $\EE[A]$ of output tokens produced
in that round. Its reciprocal, $1/T_{\mathrm{eff}}^{(\cdot)}$, is the
corresponding per-request token throughput. For co-located SD,
\begin{equation}
    T_{\text{eff}}^{\text{coloc}} = \frac{\gamma\,\tdraft + \tverify}{\EE[A]}.
    \label{eq:coloc_eff}
\end{equation}
We adopt the \emph{memory-bound assumption} $\tverify \approx \tar$: a single
verification pass over the $\gamma$ proposed tokens is dominated by HBM-to-SM
weight transfer rather than FLOPs, so it runs in roughly the same wall-clock as
one autoregressive step.

Published single-node measurements support this round-time decomposition.
Yan et al.~\citep{yan2025decoding} report the Tokens-Accepted Rate, or TAR, with
TAR$\equiv\EE[A]$; the time to generate $\gamma$ draft tokens; the
time for the target to verify them; and measured throughput for $350{+}$
experiments on OPT-66B and LLaMA-65B targets using 4$\times$A100 GPUs. Substituting
their measured draft and verification times for $\gamma\tdraft$ and
$\tverify$ in~\eqref{eq:coloc_eff} gives the predicted throughput
$1/T_{\mathrm{eff}}^{\mathrm{coloc}}$, which closely matches their
measurements. This empirical check uses the measured $\tverify$ directly and
does not require the additional approximation $\tverify\approx\tar$.

\subsection{DSD (synchronous)} When the draft model lives on the edge and the
verifier in the cloud, each round traverses the network: the payload-independent
round-trip $\rtt$, which includes propagation plus processing and is ping-measurable, plus a
payload-dependent transmission time
\begin{equation}
    T_{\mathrm{tx}}(\gamma) = \gamma \cdot b / R,
    \label{eq:ttx}
\end{equation}
where $b$ is the per-draft payload size and $R$ is the link bandwidth. The
payload size $b$ is set by what the verifier needs in order to check a draft
token, which in turn depends on the decoding mode. Under \emph{greedy}
verification the cloud only checks whether each proposed token matches the
target's argmax, so the edge sends bare token IDs: $b$ is a few bytes and the
transmission time is negligible relative to the round-trip,
$T_{\mathrm{tx}}(\gamma) \ll \rtt$, at common WAN bandwidths. Exact
\emph{distribution-preserving} verification, by contrast, must compare the target
and draft probabilities of each proposed token and resample rejected positions
from the residual $(p-q)_+$, which requires probability information over the
vocabulary. Naive logit-based protocols therefore send $b \approx |V|\,b_{\mathrm{prob}}$
bytes per draft token, where $|V|$ is the vocabulary size and $b_{\mathrm{prob}}$
the bytes per probability value. At FP16 or BF16, this is $2$ bytes per value; the payload is larger by
orders of magnitude, and $T_{\mathrm{tx}}(\gamma)$ can then dominate $\rtt$ at
lower bandwidths. Communication-efficient distributed SD introduced by
DSSD~\citep{ning2025dssd} uploads only token IDs and scalar draft probabilities,
and defers the full distribution to a downlink transmission only on rejection,
effectively reducing $b$. We use \emph{low-transmission-overhead regime} for
this communication pattern when its transfer time is small enough to omit from
the per-round latency model, written $T_{\mathrm{tx}}(\gamma)\approx0$.  
The payload is nonzero, but its expected transfer time is small relative to the
round-trip and compute terms. 
Separating the payload-independent round-trip from the payload-dependent transmission
time, the per-token effective time is:

\begin{equation}
    T_{\text{eff}}^{\text{dsd}} = \frac{\gamma\,\tdraft + \rtt + T_{\mathrm{tx}}(\gamma) + \tverify}{\EE[A]}.
    \label{eq:dist_eff}
\end{equation}

\subsection{DSD (pipelined)} Asynchronous DSD systems such as
PicoSpec~\citep{zhang2026picospec} overlap the next round's edge drafting
with the current round's network transfer and cloud verification. Thus, the
sequential sum in~\eqref{eq:dist_eff} becomes the maximum of the drafting and
cloud-side branches. A rejection can invalidate tokens drafted speculatively
for the next round, so we use a speculative-waste fraction $w\in[0,1]$ to
represent the resulting extra draft work:

\begin{equation}
    T_{\text{eff}}^{\text{pipe}} = \frac{\max\big((1+w)\,\gamma\,\tdraft,\;\rtt + T_{\mathrm{tx}}(\gamma) + \tverify\big)}{\EE[A]}.
    \label{eq:pipe_eff}
\end{equation}

Here $w=0$ represents perfect overlap without wasted pre-drafts, while larger
$w$ accounts for pipeline flushes caused by rejection.

Under the memory-bound assumption $\tverify \approx \tar$, synchronous DSD
beats the cloud-AR baseline iff $T_{\text{eff}}^{\text{dsd}} < \tar$,
equivalently the network budget is bounded by:
\begin{equation}
    \rtt + T_{\mathrm{tx}}(\gamma) < \tar\,\EE[A] - \gamma\,\tdraft - \tverify \triangleq \rtt_{\max}
    \label{eq:rtt_max}
\end{equation}

In the rest of the paper, we assume $T_{\mathrm{tx}}(\gamma)\approx0$ in
numerical examples, while in formal comparisons we retain the term.


\section{Position: The Gain Window Is Narrow}
\label{sec:position}
We evaluate DSD against the two baselines available to inference service providers,
namely \textbf{co-located SD}, where both models run on the same server, and
\textbf{cloud AR}, where only the target model runs and there is no draft model.
For each baseline, we examine DSD via closed-form inequalities in a
\emph{per-request comparison}: one active decoding request under fixed
service times $\tdraft$, $\tverify$, and $\tar$, excluding queueing delays and
resource contention from other requests. We compare stream-level latency,
per-output compute, and inter-device communication, together with two
configuration-level properties: total model-weight memory and API feasibility
or cost. Under heavy load, batching consumes
more accelerator capacity and speculative computation may cease to improve
goodput; serving systems may therefore shorten or disable speculation~\citep{liu2024turbospec}.

We also assume that the draft and target models are the same across
all configurations, so that the acceptance rate $\alpha$ and hence the expected
number of accepted tokens $\EE[A]$ are consistent.

\subsection{Comparison with Co-Located SD}


We argue that co-located SD matches or outperforms synchronous DSD across all
dimensions of this comparison.

\begin{proposition}
\label{prop:coloc-dom}
Assume that the same draft and target models, acceptance behavior, and fixed
contention-free service times apply in both configurations, and that the server
can hold both models. Then co-located SD matches or beats synchronous DSD on
every comparison dimension: latency, per-output compute (FLOPs), total
model-weight memory, and inter-device communication. 
\end{proposition}

\begin{proof}
We establish each metric in turn.


\emph{(i) Latency.} Equation~\eqref{eq:dist_eff} adds the nonnegative terms
$\rtt$ and $T_{\mathrm{tx}}(\gamma)$ to~\eqref{eq:coloc_eff}, so
$T_{\mathrm{eff}}^{\mathrm{dsd}} \ge T_{\mathrm{eff}}^{\mathrm{coloc}}$, strictly
for any $\rtt > 0$. 


\emph{(ii) Compute.} Per-round FLOPs equal $\gamma\,C_{\mathrm{draft}} +
C_{\mathrm{verify}}$ in both configurations, independent of model location.

\emph{(iii) Memory.} Total model-weight memory equals $|M_t| + |M_s|$ in both
configurations; only the draft's residency differs.

\emph{(iv) Communication.} Co-located SD samples the probability distributions $p$
and $q$ from shared GPU memory and incurs zero inter-device communication.
DSD transmits at least $\gamma$ token IDs and one accept-count per
round.

By (i) and (iv), co-located SD is strictly better than synchronous DSD on latency and
communication; by (ii) and (iii), it ties on compute and memory. 
Hence synchronous DSD has no strict per-request or configuration-level
advantage.

\end{proof}


\subsection{Against Cloud AR}

Cloud AR is the baseline DSD must beat to
justify its existence. We work through latency, compute, API feasibility,
memory, and hypothetical API cost in turn.



\begin{proposition}[Synchronous DSD latency bound]
\label{prop:lat-ar}
Assume $\tverify=\tar$ and
$T_{\mathrm{tx}}(\gamma)=\gamma b/R$. Synchronous DSD has lower mean
wall time per output token than cloud AR if and only if
\begin{equation}
\rtt < \tar(\EE[A]-1)-\gamma\!\left(\tdraft+\tfrac{b}{R}\right).
\label{eq:rtt_exact}
\end{equation}
For $0\le\alpha<1$, this exact condition implies the necessary bound
\begin{equation}
\rtt < \frac{\alpha\tar}{1-\alpha}
         -\gamma\!\left(\tdraft+\tfrac{b}{R}\right).
\label{eq:rtt_upper}
\end{equation}
\end{proposition}
\begin{proof}
Substitute $\tverify=\tar$ and
$T_{\mathrm{tx}}(\gamma)=\gamma b/R$ into
Eq.~\eqref{eq:rtt_max}. For $\alpha<1$,
$\EE[A]-1=\alpha(1-\alpha^\gamma)/(1-\alpha)
\le\alpha/(1-\alpha)$, which gives Eq.~\eqref{eq:rtt_upper}.
\end{proof}


Equation~\eqref{eq:rtt_exact} is the exact break-even condition under the
stated timing model. Equation~\eqref{eq:rtt_upper} is a looser 
bound expressed only through $\alpha$. In the low-transmission-overhead regime
defined in \S\ref{sec:background}, we set $b/R\approx0$ in
Eq.~\eqref{eq:rtt_exact}.

Table~\ref{tab:scenarios} illustrates the exact break-even condition
in~\eqref{eq:rtt_max} for $\gamma=5$, $\tdraft=10$~ms,
$\tverify=\tar$, and the low-transmission-overhead regime defined in~\S\ref{sec:background}. The feasible window contracts as the
cloud target becomes faster or the acceptance rate
falls.

\begin{table}[t]
\centering
\small
\caption{Break-even RTT (ms) from~\eqref{eq:rtt_max}. A dash denotes
$\rtt_{\max}<0$: DSD is slower than cloud AR even at zero RTT, so no feasible
break-even RTT exists.}
\label{tab:scenarios}
\begin{tabular}{@{}lcccc@{}}
\toprule
Cloud AR $\tar$ & $\alpha=0.5$ & $\alpha=0.7$ & $\alpha=0.85$ & $\alpha=0.9$ \\
\midrule
Slow, $100$~ms      & $47$ & $144$ & $265$ & $319$ \\
Standard, $50$~ms   & ---  & $47$  & $108$ & $134$ \\
Fast, $30$~ms       & ---  & $8$   & $45$  & $61$  \\
Very fast, $20$~ms  & ---  & ---   & $13$  & $24$  \\
\bottomrule
\end{tabular}
\end{table}

At an illustrative 4G RTT of ${\sim}60$~ms, the $\tar=100$~ms target
requires roughly $\alpha\ge0.7$, the $\tar=50$~ms target requires roughly
$\alpha\ge0.85$, and the faster targets have no feasible point in the table.
At an illustrative cross-region RTT of ${\sim}80$~ms, targets with
$\tar\le30$~ms are infeasible throughout the displayed acceptance range, both
are operating points instead of direct measurements. Reported 4G access latency
is $30$--$40$~ms before account for the wide-area path to a cloud
region~\citep{varghese2021edge}, and mobile-network studies report end-to-end
RTTs of approximately $25$~ms for 4G LTE and $37$--$150$~ms for 5G, depending on
the operator and location~\citep{song2023london5g}.

Published DSSD measurements exhibit the same failure mode~\citep{ning2025dssd}.
That paper defines speedup as protocol throughput divided by target-only LLM
throughput. At 50~ms non-transmission delay, 10~Mbps bandwidth, and
$\gamma=8$, its predecessor DSD method achieves only $43\%$ of the throughput
of its OPT-6.7B cloud-AR baseline. DSSD improves this to $2.19\times$ for
OPT-6.7B and $1.62\times$ for OPT-13B by transmitting a full vocabulary
distribution only on rejection. These results show that communication design
can move a DSD system into or out of the cloud-AR break-even window; because
the baseline is cloud AR rather than co-located SD, they do not address
Prop.~\ref{prop:coloc-dom}.

\begin{proposition}[Compute]
\label{prop:compute-ar}
Let $c = C_{\mathrm{draft}}/C_{\mathrm{AR}}$ be the per-token FLOP ratio of the
draft model to cloud AR, and assume verifying $\gamma$ tokens costs
$C_{\mathrm{verify}} = \gamma\,C_{\mathrm{AR}}$ FLOPs. DSD uses strictly more
FLOPs per output token than cloud AR whenever
\begin{equation}
    \EE[A] \;=\; \frac{1 - \alpha^{\gamma+1}}{1 - \alpha} \;<\; \gamma(1+c).
    \label{eq:compute-cond}
\end{equation}
This holds for all $\alpha$ once $c \ge 1/\gamma$,
and fails only in the corner case $c < 1/\gamma$ with $\alpha \to 1$.
\end{proposition}
\begin{proof}
A DSD round costs $\gamma\,C_{\mathrm{draft}} + C_{\mathrm{verify}} =
\gamma(1+c)\,C_{\mathrm{AR}}$ FLOPs and yields $\EE[A]$ tokens. Therefore, its per-token
FLOPs exceed cloud AR's $C_{\mathrm{AR}}$ iff $\gamma(1+c) > \EE[A]$. Using
$\EE[A] \le \gamma + 1$, a sufficient condition for \eqref{eq:compute-cond} to hold
is $\gamma + 1 < \gamma (1+c)$, which results in $c \ge 1/\gamma$. 
\end{proof}

\begin{remark}
The corner case is operationally empty. It requires $c < 1/\gamma$ together with
$\alpha$ near $1$ --- e.g. at $\gamma = 5$, and $c = 0$ it needs $\alpha \approx 0.93$,
with higher $\alpha$ demanded as $c$ rises toward $1/\gamma$. But $c < 1/\gamma$
means the draft is nearly free to run \emph{on the server} as well, so co-located
SD's draft cost is also small and DSD offloads negligible compute even as it
``wins.'' DSD's FLOP disadvantage therefore holds in any realistic deployment.
\end{remark}

\begin{remark}[Memory]
During inference, cloud AR loads only the target model weights into
cloud memory, with footprint $|M_t|$. DSD loads the same target model weights
into cloud memory and also loads the draft model weights on the edge device,
giving a system-wide model-weight footprint of $|M_t|+|M_s|$. Thus DSD has no
system-wide model-weight memory advantage over cloud AR; it changes where the
draft model resides, not whether the draft model exists.
\end{remark}

\begin{remark}[API cost --- infeasibility]
\label{rem:api-ar}
Distribution-preserving SD against a closed-source commercial API as the
verifier requires scoring client-proposed draft tokens, where the client needs
to acquire $p(x_i \mid x_{<i})$ for each proposed draft token $x_i$ from the
target model. Mainstream closed-source APIs generally expose logprobs for
generated tokens only, not for arbitrary tokens proposed by the client. Thus
DSD is not implementable through today's closed-source APIs without a
custom verify-only endpoint. Even if such an endpoint were offered, DSD would
have no natural API-cost advantage over provider-side co-located SD: a standard
endpoint already charges $p_{\mathrm{out}}$ per generated token regardless of
whether the provider uses SD internally, while DSD would add a separate
verification cost. Pricing that cost low enough to undercut the standard
output service would cannibalize provider revenue.
\end{remark}

\begin{remark}[Cost for hypothetical verifier APIs]
\label{prop:hypo-api}
Suppose a provider offered a verify-only endpoint charging a per-call
verification fee $F_{\mathrm{ver}}(\gamma)$ (absorbing any bundled
correction/bonus), with $p_{\mathrm{in}}$ and $p_{\mathrm{out}}$ the per-token
input and output charges. A DSD round then costs $\gamma\,p_{\mathrm{in}} +
F_{\mathrm{ver}}(\gamma)$ for $\EE[A]$ tokens, so DSD is cheaper than cloud AR
iff
\begin{equation}
    \EE[A] > \frac{\gamma\,p_{\mathrm{in}} + F_{\mathrm{ver}}(\gamma)}{p_{\mathrm{out}}}.
    \label{eq:hypo-breakeven}
\end{equation}
The right-hand side of~\eqref{eq:hypo-breakeven} is the cost of one
verification round, expressed in units of the ordinary output-token price. DSD
is economical only when that normalized round cost is smaller than the expected
number of output tokens produced by the round. The conclusion therefore depends
on a provider's hypothetical verification cost: charging for all proposed
tokens makes break-even difficult, whereas a low flat verification fee makes it
possible at moderate acceptance rates.
\end{remark}

\subsection{Multi-Tenant Server Capacity}
The single-client per-token times above describe one edge--cloud pair. We now
turn to the server side: at a common required per-client output rate, how many
concurrent clients can the cloud sustain under each protocol?

\begin{proposition}[Server-capacity gain under multi-tenant serving]
\label{prop:throughput}
Let $N_{\mathrm{ar}}$, $N_{\mathrm{coloc}}$, and $N_{\mathrm{dsd}}$ be the
maximum numbers of clients a saturated, work-conserving server can sustain
under cloud AR, co-located SD, and DSD at a common per-client output rate $r$.
Assume \emph{cross-client overlap}: the server fills each client's edge-drafting
and network-transit phase with verification work for other clients, as in
batched cloud inference. Then
\begin{equation}
    N_{\mathrm{ar}} : N_{\mathrm{coloc}} : N_{\mathrm{dsd}}
    \;=\; 1 \;:\;
    \frac{\EE[A]\,\tar}{\gamma\,\tdraft+\tverify}
    \;:\; \frac{\EE[A]\,\tar}{\tverify}.
    \label{eq:capacity-ratio-general}
\end{equation}
Under the memory-bound condition $\tverify \approx \tar$, this reduces to
\begin{equation}
    N_{\mathrm{ar}} : N_{\mathrm{coloc}} : N_{\mathrm{dsd}}
    \;\approx\; 1 \;:\; \frac{\EE[A]}{1 + \gamma\,\tdraft/\tverify} \;:\; \EE[A],
    \label{eq:capacity-ratio}
\end{equation}
with $\EE[A]$ from~\eqref{eq:expected}. Because saturated aggregate throughput
is $Nr$ at the common rate, it scales by the same ratios.
\end{proposition}

\begin{proof}
Normalize the server's available occupancy to one. A cloud-AR client producing
$r$ tokens per second consumes occupancy $r\tar$, so
$N_{\mathrm{ar}}=1/(r\tar)$. A co-located SD round occupies the server for
$\gamma\,\tdraft+\tverify$ and yields $\EE[A]$ tokens, so each client consumes
$r(\gamma\,\tdraft+\tverify)/\EE[A]$ and
$N_{\mathrm{coloc}}=\EE[A]/[r(\gamma\,\tdraft+\tverify)]$. Under DSD, drafting
moves to the edge and a round occupies the server only for $\tverify$, giving
$N_{\mathrm{dsd}}=\EE[A]/(r\tverify)$. The DSD occupancy excludes edge drafting,
network transit, and transmission because cross-client overlap lets the server
process other clients during those phases. Normalizing the three client counts
by $N_{\mathrm{ar}}$ gives~\eqref{eq:capacity-ratio-general}; substituting
$\tverify\approx\tar$ gives~\eqref{eq:capacity-ratio}.
\end{proof}

Note that the overlap assumption in Prop. \ref{prop:throughput} is
\emph{across} clients --- batching independent clients' verifications --- and is
weaker than the \emph{within}-client pipelining of \S\ref{sec:pipelining}, which
fails in the WAN regime; it requires only enough concurrent clients to keep
verification work available during each idle phase. With partial overlap, the
DSD-over-co-located capacity factor $1+\gamma\tdraft/\tverify$ should be read
as an upper bound: realized gains degrade when the server cannot fully fill a
client's edge-drafting or network phase with verification work from other
clients. With a single client that condition is empty: the round takes
$\gamma\,\tdraft + \tverify$ under co-located SD versus $\gamma\,\tdraft +
\rtt + T_{\mathrm{tx}}(\gamma) + \tverify$ under DSD, so DSD merely produces the
same $\EE[A]$ tokens \emph{more slowly}.

\begin{remark}
Let $\rho=\tverify/\tar$. Equation~\eqref{eq:capacity-ratio-general} gives
$N_{\mathrm{dsd}}/N_{\mathrm{ar}}=\EE[A]/\rho$. The specialization
$N_{\mathrm{dsd}}/N_{\mathrm{ar}}\approx\EE[A]$ therefore requires
$\rho\approx1$. MagicDec shows that this condition depends on batch size and
context length: at large batch with short contexts, verification becomes
compute-bound and $\rho$ rises, reducing or reversing DSD's capacity advantage
over cloud AR~\citep{sadhukhan2025magicdec}. In the idealized compute-bound
limit $\rho\approx\gamma$, the ratio approaches $\EE[A]/\gamma$. DSD then
supports more clients than AR only in the high-acceptance corner
$\EE[A]>\gamma$, and even perfect acceptance gives at most
$(\gamma+1)/\gamma$. Beyond a hardware- and model-dependent critical sequence
length, KV-cache loading again dominates and $\rho\approx1$ can hold even at
large batch~\citep{sadhukhan2025magicdec}.

The comparison with co-located SD does not require
$\tverify\approx\tar$:
$N_{\mathrm{dsd}}/N_{\mathrm{coloc}}
=1+\gamma\,\tdraft/\tverify$. This factor approaches $1$ specifically when
$\tverify/(\gamma\,\tdraft)$ grows, rather than merely whenever $\tverify$
grows.
\end{remark}

\begin{remark}
Under the memory-bound case~\eqref{eq:capacity-ratio}, two readings
matter. First, $N_{\mathrm{dsd}}/N_{\mathrm{ar}} = \EE[A]$ while
$N_{\mathrm{dsd}}/N_{\mathrm{coloc}} = 1 + \gamma\,\tdraft/\tverify$:
co-located SD alone already captures most of the capacity gain
over cloud AR, and distributing the draft adds only the
$(1 + \gamma\,\tdraft/\tverify)$ factor on top. Second, this distribution
factor requires enough concurrent clients for cross-client overlap; with a
single client the condition is absent, and DSD merely adds network latency over
co-located SD, as stated in Prop.~\ref{prop:coloc-dom}.
\end{remark}

Published multi-tenant results support both comparisons in this remark.
SLED~\citep{li2025sled} reports a $2.2\times$ system-throughput improvement
for its DSD framework over cloud AR, consistent with the memory-bound
$N_{\mathrm{dsd}}/N_{\mathrm{ar}}\approx\EE[A]$ branch at typical
$\EE[A]\approx 2$. SpecEdge~\citep{park2025specedge} instead compares
against co-located SD: offloading tree drafting raises server
throughput by $2.22\times$, illustrating the separate
$N_{\mathrm{dsd}}/N_{\mathrm{coloc}}=1+\gamma\tdraft/\tverify$ branch at
a draft-heavy operating point. CoSine~\citep{gao2025cosine} gives complementary
datacenter evidence that separating draft and verification resource demands can
improve throughput. These measurements support the direction of the capacity
result; they do not establish a universal numerical factor because the systems
use different models, batching policies, drafting mechanisms, and interconnect
regimes.

\begin{remark}
For inference service providers, co-located SD has lower latency and
communication than DSD and matches DSD on FLOPs and memory by
Prop.~\ref{prop:coloc-dom}; DSD reduces latency over cloud AR only in the
bounded window of Prop.~\ref{prop:lat-ar}.
Closed-API end users instead face the API limitation in Rem.~\ref{rem:api-ar}.
Under multi-tenant serving, speculation alone already gives co-located SD a
large capacity gain over cloud AR --- the server verifies $\EE[A]$ tokens per
forward pass instead of generating one. DSD's additional capacity benefit over co-located SD is the
$(1 + \gamma\,\tdraft/\tverify)\times$ more clients supported by offloading
drafting from the server. Prop.~\ref{prop:throughput} shows that this factor is governed entirely by the
on-server draft cost $\gamma\,\tdraft/\tverify$, significant only when drafting
is expensive, and shrinking as native multi-token
prediction cheapens on-server drafting; see \S\ref{sec:pipelining}. For a
single request, synchronous DSD is slower than co-located SD, and pipelined
DSD remains no faster in the WAN regime of Prop.~\ref{prop:pipe-coloc}.
Under sufficient cross-client overlap, DSD's modeled advantage is instead
greater server capacity.
\end{remark}

\section{What Pipelining Changes (And What It Doesn't)}
\label{sec:pipelining}
The latency part of Prop.~\ref{prop:coloc-dom} assumes synchronous DSD.
A natural objection is that an asynchronous protocol could hide the RTT behind
edge drafting and thereby overturn DSD's latency disadvantage relative to
co-located SD. Equations~\eqref{eq:dist_eff} and~\eqref{eq:pipe_eff}
make this distinction explicit: synchronous DSD adds edge drafting, network,
and verification time, whereas pipelined DSD takes the maximum of the
overlapped drafting and cloud-side branches. Most DSD systems we discuss,
including DSSD~\citep{ning2025dssd} and DSD~\citep{yu2025dsd}, use the
synchronous structure already covered above. PicoSpec~\citep{zhang2026picospec}
and PipeSD~\citep{han2026pipesd} are recent asynchronous cases, so the following proposition tests whether this
stronger pipelined variant changes the latency ordering in the edge-cloud WAN
regime. Here $T_{\mathrm{round}}^{\mathrm{pipe}}$ and
$T_{\mathrm{round}}^{\mathrm{coloc}}$ denote per-round wall times,
$T_{\mathrm{round}}^{X}=T_{\mathrm{eff}}^{X}\EE[A]$; $w$ is the speculative-waste
fraction from~\eqref{eq:pipe_eff}; and the low-transmission-overhead regime is
defined in~\S\ref{sec:background}. The condition $\rtt\ge\gamma\tdraft$
describes the case where one cloud round trip takes at least as long as drafting
one full speculation round on the edge. In this case, the pipeline cannot hide
the round trip inside edge drafting; the network branch is at least as large as
the co-located drafting branch before verification is added.

\begin{proposition}[DSD with pipelining does not dominate co-located SD in the WAN regime]
\label{prop:pipe-coloc}
In the low-transmission-overhead model obtained by setting
$T_{\mathrm{tx}}(\gamma)=0$ and for $\rtt \ge \gamma \tdraft$,
$T_{\mathrm{round}}^{\mathrm{pipe}} \ge
T_{\mathrm{round}}^{\mathrm{coloc}}$
\end{proposition}

\begin{proof}
Write $T_{\mathrm{round}}^{X} \triangleq
T_{\mathrm{eff}}^{X} \EE[A]$ for per-round wall time under
configuration $X$. From~\eqref{eq:pipe_eff} and~\eqref{eq:coloc_eff}, after
setting $T_{\mathrm{tx}}(\gamma)=0$,
\begin{align}
T_{\mathrm{round}}^{\mathrm{pipe}}
  &= \max((1+w)\gamma\tdraft,\; \rtt + \tverify), \\
T_{\mathrm{round}}^{\mathrm{coloc}}
  &= \gamma\tdraft + \tverify .
\end{align}
If $\rtt \ge \gamma\tdraft$, then the cloud-side branch satisfies
$\rtt + \tverify \ge \gamma\tdraft + \tverify =
T_{\mathrm{round}}^{\mathrm{coloc}}$. Since
$T_{\mathrm{round}}^{\mathrm{pipe}}$ is the maximum of this branch and the
draft branch, $T_{\mathrm{round}}^{\mathrm{pipe}} \ge
T_{\mathrm{round}}^{\mathrm{coloc}}$.
\end{proof}

PipeSD~\citep{han2026pipesd} overlaps edge draft generation and communication
using a dynamic-programming token-batch scheduler and a dual-threshold
verification trigger. Its $1.16$--$2.16\times$ speedups are measured against
other DSD baselines such as Vanilla, HSL, and
EdgeLLM, rather than against co-located SD. Therefore they support the claim that
pipelining improves DSD implementations without changing 
Prop.~\ref{prop:pipe-coloc}.

Using the same illustrative values as Table~\ref{tab:scenarios},
$\gamma = 5$ and $\tdraft = 10$~ms, we have $\gamma\tdraft = 50$~ms. In
this case, the condition
$\rtt \ge \gamma\tdraft$ in Prop.~\ref{prop:pipe-coloc} holds at both
illustrative operating points of \S\ref{sec:position}, the 4G RTT of
${\sim}60$~ms and the cross-region RTT of ${\sim}80$~ms.

Low RTT is necessary but not sufficient for pipelined DSD to outperform
co-located SD. From Eqs.~\eqref{eq:pipe_eff} and~\eqref{eq:coloc_eff},
pipelined DSD is strictly faster iff
\[
\rtt + T_{\mathrm{tx}}(\gamma) < \gamma\tdraft
\quad\text{and}\quad
w\gamma\tdraft < \tverify.
\]
The first condition requires communication to fit within one edge drafting
round. The second requires the cost of wasted pre-drafting to remain below the
verification time. Nearby WiFi, metro-edge, or favorable 5G links may satisfy
these conditions. Prop.~\ref{prop:pipe-coloc} rules out a latency advantage
whenever $\rtt \ge \gamma\tdraft$.

Two recent systems exploit the $\rtt < \gamma\tdraft$ regime, but only
inside the datacenter. CoSine~\citep{gao2025cosine} splits the matrix-matrix
multiplication and the matrix-vector multiplication parts of
LLM inference computation across heterogeneous accelerators connected by a 10~Gbps
datacenter network. DSI~\citep{timor2024dsi} runs speculative branches in
parallel across multiple GPUs and proves that this datacenter-parallel schedule
is no slower than non-speculative decoding, and no slower in expectation than
standard speculative decoding. Both are datacenter-internal parallel-inference
systems operating at sub-millisecond RTT. They are complementary to our scope,
rather than latency targets for our WAN comparison, because CoSine relies on a
fast datacenter interconnect and DSI on abundant parallel accelerators. Neither
removes the WAN round trip. CoSine nevertheless supports our capacity argument
by showing that drafting and verification can contend for different accelerator
resources under load, and that separating their execution across heterogeneous
resources can improve serving throughput.

SpecEdge~\citep{park2025specedge} is consistent with the
$\rtt < \gamma\tdraft$ branch for nearby WiFi, metro-edge, or favorable 5G
links. In Prop.~\ref{prop:pipe-coloc}, $\gamma\tdraft$ is the wall-clock time
spent drafting one speculation round. SpecEdge uses tree-based drafting with
adaptive depth, so the analogous quantity is the edge draft-phase time, estimated
as draft depth times per-pass draft latency. In its 32B/1.7B calibration, verification averages
$94.2$~ms and each draft forward pass takes about $11$~ms. SpecEdge therefore
uses depths $7$, $5$, and $4$ at RTTs of $15$, $40$, and $50$~ms, giving edge
draft times of about $77$, $55$, and $44$~ms. Thus the low-RTT operating points
satisfy the same timing condition as $\rtt < \gamma\tdraft$, whereas the
$50$~ms point is already at the boundary where RTT is no longer hidden by edge
drafting. This matches its reported latency behavior: at low RTT, SpecEdge
reports $36.5$~ms inter-token latency versus $42.4$~ms for co-located SD, while
its sensitivity analysis places the crossover near $50$~ms RTT; at $65$~ms RTT,
SpecEdge reaches $44.5$~ms and is slower than the co-located baseline.

Two effects narrow DSD's remaining capacity advantage. First,
MagicDec observes that verification can become compute-bound in large-batch,
short-context settings~\citep{sadhukhan2025magicdec}. In that regime,
$\tverify/\tar$ rises, so Prop.~\ref{prop:throughput} predicts a smaller DSD
capacity gain over cloud AR. Second, native
multi-token prediction in frontier models such as DeepSeek-V3~\citep{deepseek2024v3}
and Qwen3-Next~\citep{qwen2025qwen3next} cheapens co-located proposal
generation. For the comparison with co-located SD, these target-side proposal
heads replace the separate draft-model cost $\gamma\tdraft$ with a smaller
proposal overhead on the target. As that overhead falls, the capacity gain from
moving draft work to the edge moves closer to no gain.


\section{Conclusion}
\label{sec:conclusion}
Against co-located SD, synchronous DSD has higher per-request latency and
communication with matching per-output compute and total model-weight memory
by Prop.~\ref{prop:coloc-dom}. Under the low-transmission-overhead model,
pipelined DSD also cannot reduce latency when
$\rtt \ge \gamma\tdraft$, by Prop.~\ref{prop:pipe-coloc}. Against cloud AR,
DSD reduces latency only within the bounded, $\tar$-dependent window of
Prop.~\ref{prop:lat-ar}; see Table~\ref{tab:scenarios}.
DSD is also infeasible against closed-source target LLM APIs.
The main case for DSD is multi-tenant capacity: under cross-client
overlap, offloading draft compute lets a saturated server sustain
$(1 + \gamma\,\tdraft/\tverify)\times$ more clients at the same per-client
output token rate by Prop.~\ref{prop:throughput}. Thus DSD should be evaluated primarily by multi-tenant capacity and server
throughput, rather than by single-request latency alone. Latency claims should report the RTT range where they hold and whether
the baseline is cloud AR or co-located SD. DSD's remaining capacity
advantage becomes smaller
in regimes where verification is compute-bound, as observed for large-batch,
short-context settings, and when native multi-token prediction is adopted in
target models that reduce co-located drafting cost.

Our reframing implies three reporting practices. First, \emph{report the
break-even RTT of Prop.~\ref{prop:lat-ar} together with its parameter setting},
so a reader can tell whether a system applies to their deployment. Second,
\emph{sweep $(\alpha, \rtt, \gamma)$ at
several target speeds $\tar$ rather than reporting a single operating point}: the
viable region is a surface, not a point, and single-condition speedups are easily
cherry-picked. Third, \emph{report multi-tenant server throughput, not just
single-user speedup}, and \emph{specify the communication protocol} (greedy
vs.\ logit-based), whose payload differs by orders of magnitude and sets the
latency window.
Throughout, \emph{co-located SD---not cloud AR---is the baseline to beat}
by Prop.~\ref{prop:coloc-dom}.

We leave to future work several deployment-level concerns where DSD may add value
but where no controlled comparison yet exists: user-prompt privacy (standard DSD
still sends the prefix to the cloud, though the architecture admits encrypted or
partial-prefix variants), per-user personalized draft models at scale (where
co-located SD would require the server to hold $N$ draft models), perceived
time-to-first-token from showing unverified drafts, and graceful degradation
under intermittent connectivity. These are legitimate directions, but orthogonal
to the latency claims the literature currently makes.

\bibliographystyle{IEEEtran}
\bibliography{references}

\end{document}